\documentclass[11pt,english]{article}
\usepackage[T1]{fontenc}
\usepackage[latin9]{inputenc}
\usepackage{amsmath}
\usepackage{amssymb}

\makeatletter
\usepackage{amsfonts}

\newcommand{\detstar}{\det\,\!\!^{\ast}}

\providecommand{\U}[1]{\protect\rule{.1in}{.1in}}
\newtheorem{theorem}{Theorem}

\newtheorem{corollary}[theorem]{Corollary}

\newtheorem{lemma}[theorem]{Lemma}

\newtheorem{proposition}[theorem]{Proposition}

\newenvironment{proof}[1][Proof]{\noindent\textbf{#1.} }{\ \rule{0.5em}{0.5em}}

\makeatother

\usepackage{babel}

\begin{document}

\title{Complexity and heights of tori}

\author{Gautam Chinta \and Jay Jorgenson \and Anders Karlsson
\footnote{The first and second authors acknowledge support provided by grants from
the National Science Foundation and the Professional Staff Congress of the City University
of New York.  The third author received support from SNSF grant 200021$\_$132528$/$1.}}
\maketitle
\begin{abstract}
We prove detailed asymptotics for the number of spanning trees, called
complexity, for a general class of discrete tori as the parameters
tend to infinity. The proof uses in particular certain ideas and
techniques from an earlier paper \cite{CJK10}. Our asymptotic formula
provides a link between the complexity of these graphs and the height
of associated real tori, and allows us to deduce some corollaries
on the complexity thanks to certain results from analytic number theory.
In this way we obtain a conjectural relationship between complexity
and regular sphere packings.
\end{abstract}

\section{Introduction}

The number of spanning trees $\tau(G)$, called \emph{the} \emph{complexity},
of a finite graph $G$ is an invariant which is of interest in several
sciences: network theory, statistical physics, theoretical chemistry,
etc. Via the well-known matrix-tree
theorem of Kirchoff, the complexity equals
the determinant of the combinatorial Laplacian $\Delta_{G}$ divided
by the number of vertices.

For compact Riemannian manifolds $M$ there is an analogous invariant
$h(M)$, \emph{the} \emph{height}, defined as the negative of the logarithm
of the zeta-regularized determinant of the Laplace-Beltrami operator,
and which is of interest for quantum physics. The analogy between the
height and complexity has been commented on by Sarnak in \cite{S90}.

In statistical physics it is of interest to study the asymptotics
of the complexity, and other spectral invariants, for certain families
of graphs. Important cases to study are various subgraphs of the standard
lattice $\mathbb{Z}^{d}$. An instance of this is to study discrete
tori, corresponding to periodic boundary conditions, as the parameters
tend to infinity, see \cite{DD88}, \cite{CJK10}, and references
therein. It is shown in \cite{CJK10} that in the asymptotics of the
complexity of discrete tori, the height of an associated real torus
appears as a constant.

In the present paper we study discrete tori of a more general type,
defined as follows. Let $\Lambda$ be an invertible $r\times r$ matrix
with all entries being integers. This matrix defines a lattice
$\Lambda\mathbb{Z}^{r}$ in $\mathbb{R}^{r}.$ We associate to the group
quotient with standard generators \[
\Lambda\mathbb{Z}^{r}\backslash\mathbb{Z}^{r}\] its Cayley graph,
which we call a \emph{discrete torus}.  In other words, two elements
$x$ and $y$ in $\Lambda\mathbb{Z}^{r}\backslash\mathbb{Z}^{r}$ are
\emph{adjacent}, denoted $x\sim y,$ if they differ by $\pm1$ in
exactly one of the coordinates and equal everywhere else (everything
mod $\Lambda\mathbb{Z}^{r}$ of course).

Let
$0=\lambda_{0}<\lambda_{1}\leq...\leq\lambda_{\left|\det\Lambda\right|-1}$
be the eigenvalues of the combinatorial Laplacian
$\Delta_{\Lambda\mathbb{Z}^{r}\backslash\mathbb{Z}^{r}}$ of the
discrete torus -- see (\ref{def:combinatorial-Laplacian}) for a
definition of
$\Delta_{\Lambda\mathbb{Z}^{r}\backslash\mathbb{Z}^{r}}$.  Define
$\det^\ast\Delta_{\Lambda\mathbb{Z}^{r}\backslash\mathbb{Z}^{r}}$ to
be the product of the nonzero eigenvalues of the Laplacian: \[
\detstar\Delta_{\Lambda\mathbb{Z}^{r}\backslash\mathbb{Z}^{r}}:=\lambda_{1}\lambda_{2}...\lambda_{\left|\det\Lambda\right|-1}.\]
Note that the trivial eigenvalue is removed. We will nevertheless
refer to
$\detstar\Delta_{\Lambda\mathbb{Z}^{r}\backslash\mathbb{Z}^{r}}$ as
the determinant of the Laplacian.  For simplicity we will
mostly assume that $\det\Lambda>0$.

\begin{theorem} \label{thmmain} Let $\{\Lambda_{n}\}$ be a sequence
of $r\times r$ integer matrices.  Suppose that
$\det\Lambda_{n}\rightarrow\infty$
and $\Lambda_{n}/(\det\,\Lambda_{n})^{1/r}\rightarrow A\in SL_{r}(\mathbb{R}).$
Then as $n\rightarrow\infty,$
\[
\log\detstar\Delta_{\Lambda_{n}\mathbb{Z}^{r}
\backslash\mathbb{Z}^{r}}=c_{r}\det\Lambda_{n}+\frac{2}{r}\log\det\Lambda_{n}+
\log\detstar\Delta_{A\mathbb{Z}^{r}\backslash\mathbb{R}^{r}}+o(1)
\]
where\[
c_{r}=\log2r-\int_{0}^{\infty}e^{-2rt}(I_{0}(2t)^{r}-1)\frac{dt}{t}.\]
The definition of $\log\detstar\Delta_{A\mathbb{Z}^{r}\backslash\mathbb{R}^{r}}$
will be recalled in section 3 below.  
\end{theorem}

Our earlier paper \cite{CJK10} treats the case when the $\Lambda_{n}$
are diagonal matrices. The present paper uses several facts that are
established in that paper. Thanks to the fact that the first two terms
in the asymptotics are universal in the sense that they only depend on
$\Lambda_{n}$ via $\det\Lambda_{n}$, the theorem gives a close
connection between the complexity of certain graphs and the height of
an associated manifold. We emphasize that this is not at all obvious:
while it is true that the appropriately rescaled eigenvalues of the
discrete tori converge to the eigenvalues of
$\Delta_{A\mathbb{Z}^{r}\backslash\mathbb{R}^{r}}$, this convergence
is certainly not uniform. Moreover, the height cannot be defined as
the product of eigenvalues, there is a regularization in the
definition. An explicit expression for the height of flat tori of
volume 1,
$h(A\mathbb{Z}^{r}\backslash\mathbb{R}^{r})
:=-\log\det\Delta_{A\mathbb{Z}^{r}\backslash\mathbb{R}^{r}}$,
can be found in Theorem 2.3 of \cite{Ch97}. Deninger and L\"uck have informed
us that the constant $c_{r}$ also has an interpretation as a
determinant, namely the Fuglede-Kadison determinant of the Laplacian
on $\mathbb{Z}^{r},$ see \cite{L}. 

We turn now to a connection between our results and sphere packings.
The problem of finding the densest packing of ordinary space with
spheres of equal radii is an old one with practical importance even
in dimensions greater than 3. One type of packings is \emph{regular
sphere packings} which means that the spheres are centered at the
points of a lattice $A\mathbb{Z}^{r}$. Gauss showed that the face
centered lattice (fcc) $D_{3}$ is optimal among regular packings
in three dimensions. By work of Thue and Toth one knows that the hexagonal
lattice $A_{2}$ is densest in dimension 2. In dimenison 24 it is
known that the Leech lattice is optimal among unimodular lattices.
We refer to the book of Conway-Sloane \cite{CS99} for more information.

Conjecturally the height of $A\mathbb{Z}^{r}\backslash\mathbb{R}^{r}$
has a global minimum when $A\mathbb{Z}^{r}$ is the densest regular
sphere packing. Extremal metrics for heights has been studied in \cite{OPS88}
in dimension 2, and for tori in higher dimensions notably in \cite{Ch97}
and \cite{SS06}. In these papers, the question is phrased as the
study of the derivative of Epstein zeta functions at $s=0.$ From
this theory we can deduce the following corollaries from our main
theorem:

\begin{corollary} \label{coroptimal} Let $\{\Lambda_{n}\}$ be a sequence
of $r\times r$ integer matrices with $\det\Lambda_{n}\rightarrow\infty.$
Suppose that $\{\Lambda_{n}/(\det\,\Lambda_{n})^{1/r}\}$ belongs to a
compact subset of $SL_{r}(\mathbb{R\mathit{)}}$,
$r=2,3$, avoiding lattices equivalent to $A_{2}$, resp. $D_{3}$.
Assume that there is a sequence $\{L_{n}\}$ with $\det L_{n}=\det\Lambda_{n}$
such that $\{L_{n}/(\det\,L_{n})^{1/r}\}$ converges to $A_{2}$,
resp. $D_{3}.$
Then $L_{n}\mathbb{Z}^{r}\backslash\mathbb{Z}^{r}$
has more spanning trees than $\Lambda_{n}\mathbb{Z}^{r}\backslash\mathbb{Z}^{r}$
for all sufficiently large $n.$ \end{corollary}

\begin{corollary} \label{corestimate}Let $\Lambda_{n}$ be a sequence
of $r\times r$ integer matrices with $\det\Lambda_{n}\rightarrow\infty$
Suppose that $\{\Lambda_{n}/(\det\,\Lambda_{n})^{1/r}\}$
stays in a compact subset of $SL_{r}(\mathbb{R)}.$
For all sufficiently large $n$ we have that \[
\tau(\Lambda_{n}\mathbb{Z}^{r}\backslash\mathbb{Z}^{r})\leq\frac{\left(\det\Lambda_{n}\right)^{2/r-1}}{4\pi}\exp(c_{r}\det\Lambda_{n}+\gamma+2/r),\]
 where $\gamma$ is Euler's constant and $c_{r}$ is as in the theorem.
\end{corollary}

In the trivial case $r=1,$ this estimate gives a value close to the
truth:\[
\det\Lambda_{n}=\tau(\Lambda_{n}\mathbb{Z}^{r}\backslash\mathbb{Z}^{r})\leq1.05\det\Lambda_{n}.\]

Upper bounds for the number of spanning trees have been considered
in the combinatorics literature since 1970s at least. For regular
graphs there is a rather sharp estimate by Chung and Yau \cite{CY99}
improving on an earlier result of McKay \cite{M}.
In general, it is an open
problem to decide which simple graph on $n$ vertices and $e$ edges
has the maximal complexity. This is of interest to communication network
theory since this graph invariant appears as a measure of reliability.

It would be of interest to also go in the other direction: proving
results on the extrema of families of Epstein zeta functions via a
better understanding of the number of spanning trees of discrete tori.

\section{Spectral preliminaries for discrete tori
\label{secprelim}}

Let $\Lambda$ be an invertible $r\times r$ integer matrix$.$ This
matrix defines a lattice $\Lambda\mathbb{Z}^{r}$ in $\mathbb{R}^{r}.$
We denote by $DT(\Lambda)$ the \emph{discrete torus}, or Cayley graph
of the quotient group $\Lambda\mathbb{Z}^{r}\backslash\mathbb{Z}^{r}$
with standard generating set: two elements $x$ and $y$ in
$\Lambda\mathbb{Z}^{r}\backslash\mathbb{Z}^{r}$ are \emph{adjacent},
denoted $x\sim y,$ if they differ by $\pm1$ in
exactly one of the coordinates and equal everywhere else (everything
mod $\Lambda\mathbb{Z}^{r}$ of course).

The associated (combinatorial) \emph{Laplacian} is defined by
\begin{equation}
  \label{def:combinatorial-Laplacian}
\Delta_{DT(\Lambda)}f(x)=\sum_{y\text{ s.t. }y\sim x}(f(x)-f(y))
\end{equation}
on functions
$f:\Lambda\mathbb{Z}^{r}\backslash\mathbb{Z}^{r}\rightarrow\mathbb{R}.$

The \emph{dual lattice} $\Lambda^{\ast}\mathbb{Z}^{r}$ is as usual
all the points $v$ in $\mathbb{R}^{r}$ such that $(v,x)\in\mathbb{Z}$
for all $x\in\Lambda\mathbb{Z}^{r}$, where $(\cdot,\cdot)$ denotes
the usual scalar product. Since $\mathbb{Z}^{r}$ is self-dual and
$\Lambda\mathbb{Z}^{r}$ is a subgroup of $\mathbb{Z}^{r}$ it follows
that $\mathbb{Z}^{r}$ is a subgroup of $\Lambda^{\ast}\mathbb{Z}^{r}$.
Note that the respective indices are\[
\lbrack\mathbb{Z}^{r}:\Lambda\mathbb{Z}^{r}]=[\Lambda^{\ast}:\mathbb{Z}^{r}]=\left|\det\Lambda\right|.\]

\begin{proposition} The eigenfunctions of $\Delta_{DT(\Lambda)}$
are given by\[
f_{v}(x)=e^{2\pi i(x,v)},\]
 for each $v\in\mathbb{Z}^{r}\backslash\Lambda^{\ast}\mathbb{Z}^{r},$
with corresponding eigenvalue given by\[
\lambda_{v}=2r-2\sum_{k=1}^{r}\cos(2\pi v_{k}),\]
 where $v_{k}$ denotes the $k$th coordinate of $v.$ \end{proposition}

\begin{proof} The operator $\Delta_{DT(\Lambda)}$ is a semipositive,
symmetric matrix and hence we are looking for $\left|\det\Lambda\right|$
number of eigenfunctions and eigenvalues. The proof is a trivial calculation:\begin{align*}
\Delta e^{2\pi i(x,v)} & =2re^{2\pi i(x,v)}-\sum_{y\sim x}e^{2\pi i(y,v)}=\\
 & =\left(2r-\sum_{z\sim0}e^{2\pi i(z,v)}\right)e^{2\pi i(x,v)}.\end{align*}
\end{proof}

The \emph{heat kernel} $K^{\Lambda}(t,x):\mathbb{R}_{\geq0}\times DT(\Lambda)\rightarrow\mathbb{R}$
is the unique bounded function which satisfies\begin{align*}
\left(\Delta_{DT(\Lambda)}+\frac{\partial}{\partial t}\right)K^{\Lambda}(t,x) & =0\\
K^{\Lambda}(0,x) & =\delta_{0}(x),\end{align*}
 where $\delta_{0}(x)=1$ if $x=0$ and $0$ otherwise. The existence
and uniqueness of heat kernels in a general graph setting is established
in \cite{DM06}. Recall from e.g. \cite{CJK10} that \[
K^{\mathbb{Z}^{r}}(t,z)={\displaystyle \prod\limits _{k=1}^{r}}K^{\mathbb{Z}}(t,z_{k})\]
 for $z=(z_{k})$ and $K^{\mathbb{Z}}(t,w)=e^{-2t}I_{w}(2t)$, where
$I_{w}$ is the $I$-Bessel function of order $w.$

We have the following \emph{theta inversion formula} (cf. \cite{CJK10}).

\begin{proposition} The following formula holds for $x\in\Lambda\mathbb{Z}^{r}\backslash\mathbb{Z}^{r}$
and $t\in\mathbb{R}_{\geq0}$\[
\sum_{y\in\Lambda\mathbb{Z}^{r}}K^{\mathbb{Z}^{r}}(t,x-y)=\frac{1}{\left|\det\Lambda\right|}\sum_{\nu\in\mathbb{Z}^{r}\backslash\Lambda^{\ast}\mathbb{Z}^{r}}e^{-t\lambda_{v}}f_{v}(x).\]
 In particular,\[
\theta_{\Lambda}(t):=\left|\det\Lambda\right|\sum_{y\in\Lambda\mathbb{Z}^{r}}e^{-2rt}I_{y_{1}}(2t)...I_{y_{r}}(2t)=\sum_{\nu\in\mathbb{Z}^{r}\backslash\Lambda^{\ast}\mathbb{Z}^{r}}e^{-t\lambda_{v}}\]

\end{proposition}

\begin{proof} Since both sides of the equation satisfy the conditions
for being the heat kernel, this follows from the uniqueness of heat
kernels. The second formula is the special case $x=0$. \end{proof}

\section{Spectral preliminaries for continuous tori}

An $r$-dimensional (continuous) torus is given as a quotient of $\mathbb{R\mathrm{^{r}}}$
by a lattice $A\mathbb{Z\mathrm{^{r}}}$ where $A\in GL_{r}(\mathbb{R}).$
The metric structure and the standard (positive) Laplace-Beltrami
operator $-\sum_{i}\partial^{2}/\partial x_{i}^{2}$ on $\mathbb{R\mathrm{^{r}}}$
projects to the torus. The volume is $\left|\det A\right|.$ Let $A^{\ast}$
be the matrix defining the dual lattice $A^{*}\mathbb{Z\mathrm{^{r}}}$,
and so $A^{*}=\left(A^{-1}\right)^{t}.$ The eigenfunctions of the
Laplacian on the torus in question are $f_{v}(x)=\exp(2\pi iv^{t}x)$
where $v$ are the vectors in the dual lattice. The corresponding
eigenvalues are $\lambda_{v}=4\pi^{2}\left\Vert v\right\Vert ^{2}$
or with a different indexing: $\lambda_{m}=(2\pi)^{2}(A^{\ast}m)^{t}(A^{*}m),$
where $m$ runs through $\mathbb{Z}^{r}.$ We have the associated
theta function\[
\Theta_{A}(t)=\sum_{m\in\mathbb{Z}^{r}}
e^{-(2\pi)^{2}(A^{\ast}m)^{t}(A^{*}m)\cdot  t}.\]
The theta inversion formula, which is equivalent to the
Poisson summation formula in this case, yields \[
\Theta_{A}(t)=\frac{1}{(4\pi t)^{r/2}}
\sum_{x\in A\mathbb{Z}^{r}}e^{-\left\vert x\right\vert ^{2}/4t}.\]

The associated spectral zeta function, which in this case
also goes under the name of the Epstein zeta function, is defined
as \[
Z_{A}(s)=\sum_{m\neq0}\lambda_{m}^{-s}=
\frac{1}{(2\pi)^{2s}}\sum_{v\neq0}\frac{1}{\left\Vert
    v\right\Vert ^{2s}}.\]

Classically, one can prove the meromorphic continuation of $Z_{A}(s)$ to all $s \in \mathbb C$, showing that
its continuation is holomorphic at $s=0$.  From this, 
one defines the spectral determinant $\detstar\Delta_{A\mathbb{Z}^{r}\backslash\mathbb{R}^{r}}$ by
$$
\log\detstar\Delta_{A\mathbb{Z}^{r}\backslash\mathbb{R}^{r}} = -Z'_{A}(0).
$$

\section{Asymptotics}

Let\[
\mathcal{I}_{r}(s)=-\int_{0}^{\infty}\left(e^{-s^{2}t}e^{-2rt}I_{0}(2t)^{r}-e^{-t}\right)\frac{dt}{t}\]
 and\[
\mathcal{H}_{\Lambda}(s)=-\int_{0}^{\infty}\left(e^{-s^{2}t}\left[\theta_{\Lambda}(t)-\left|\det\Lambda\right|\cdot e^{-2rt}I_{0}(2t)^{r}-1\right]+e^{-t}\right)\frac{dt}{t}.\]
 Everything in section 3 of \cite{CJK10} carries over with only notational
changes, even though the eigenvalues are different and the theta identity
is hence somewhat different. These differences are not essentially
used.  In particular the first order term as $t\rightarrow0$ in the
trace of the heat kernel is still (in the present notation) $\left|\det\Lambda\right|\cdot e^{-2rt}I_{0}(2t)^{r}$
since it corresponds to the trivial eigenvalue. In particular the
following extension of Theorem 3.6 in \cite{CJK10} holds:

\begin{theorem} \label{thmlogderivatives}For any $s\in\mathbb{C}$
with $\operatorname{Re}(s^{2})>0,$ we have the relation\[
\sum_{\lambda_{v}\neq0}\log(s^{2}+\lambda_{\nu})=\left|\det\Lambda\right|\cdot\mathcal{I}_{r}(s)+\mathcal{H}_{\Lambda}(s).\]
 Letting $s\rightarrow0$ we have the identity\[
\log({\displaystyle \prod\limits _{\lambda_{v}\neq0}}\lambda_{\nu})=\left|\det\Lambda\right|\cdot\mathcal{I}_{r}(0)+\mathcal{H}_{\Lambda}(0).\]

\end{theorem}

Section 4 of \cite{CJK10} is an independent section on uniform bounds
on $I$-Bessel functions. We recall the following statements, slightly
adapted to the present context (keeping in mind that $I_{-y}=I_{y}$
for integers, and that $b$ may now be real):

\begin{proposition} \label{propestimates}For any $t>0$ and $b\geq0$,
there is a constant $C$ such that\[
0\leq\sqrt{b^{2}t}e^{-b^{2}t}I_{0}(b^{2}t)\leq C<1\]
 Fix $t\geq0$ and integers $y,\, n_{0}\geq0.$ Then for all $b\geq n_{0}$
we have the uniform bound\[
0\leq\sqrt{b^{2}t}\cdot e^{-b^{2}t}I_{y}(b^{2}t)\leq\left(1+\frac{y}{bn_{0}t}\right)^{-n_{0}y/2b}.\]

\end{proposition}

\begin{proposition} \label{prophkconv}Let $N(u)$ be a sequence
of positive integers parametrized by $u\in\mathbb{Z}_{+}$ such that
$N(u)/u\rightarrow\alpha>0$ as $u\rightarrow\infty.$ Then for any
$t>0$ and integer $k$, we have\[
\lim_{u\rightarrow\infty}N(u)e^{-2u^{2}t}I_{N(u)k}(2u^{2}t)=\frac{\alpha}{\sqrt{4\pi t}}e^{-(\alpha k)^{2}/4t}.\]
\end{proposition}

From now on we fix a sequence $\{\Lambda_{n}\}$ of integer matrices with
$0<\det\Lambda_{n}\rightarrow\infty$ satisfying
 \[
\frac{1}{(\det\Lambda_{n})^{1/r}}\Lambda_{n}\rightarrow A\text{ as }n\rightarrow\infty,\]
for some $A \in \textrm{SL}_{r}(\mathbb R)$.  
 From the previous propositions we will deduce the following:

\begin{proposition} \label{proppointwise}For each fixed $t>0$ we
have the pointwise convergence\[
\theta_{\Lambda_{n}}(\det(\Lambda_{n})^{2/r}t)\rightarrow\theta_{A}(t)\]
 as $n\rightarrow\infty.$ \end{proposition}

\begin{proof} 
For any $v\in\mathbb{Z}^{r}$ and $\Lambda \in \text{\rm GL}_{r}(\mathbb R)$, let
$$
{\mathbf I}_{v,\Lambda}(t) = \prod\limits_{i=1}^{r}I_{(\Lambda v)_{i}}(t)
$$ 
where $(\Lambda v)_{i}$ denotes the $i$-th component of $\Lambda v$.  
Let $u_{n}=\det(\Lambda_{n})^{1/r}$ and
$a_{i}=\left(Av\right)_{i}.$
Note that $(\Lambda v)_{i}/u_{n}\rightarrow a_{i}.$ We have\[
\theta_{\Lambda_{n}}(u_{n}^{2}t)=\sum_{v\in\mathbb{Z}^{r}}u_{n}^{r}e^{-2ru_{n}^{2}t}
{\mathbf I}_{v,\Lambda_{n}}(2u_{n}^{2}t).
\]
From Proposition \ref{prophkconv} (with $k=0$ or $\pm1$) we
have for any $t>0$ and $v\in\mathbb{Z}^{r}\ $ that\[
u^{r}_{n} e^{-2ru_{n}^{2}t}{\mathbf I}_{v,\Lambda_{n}}(2u_{n}^{2}t)\rightarrow\frac{1}{\left(\sqrt{4\pi t}\right)^{r}}e^{-a_{1}^{2}/4t}...e^{-a_{r}^{2}/4t}\]
 as $n\rightarrow\infty.$ This means that the proposition will be
proved if we can interchange the limit and the infinite sum. We show
that for fixed $t$, the sum is convergent uniformly in $u_{n}$ (or
equivalently, in $n).$

We can rewrite the sum $\theta_{\Lambda_{n}}(u_{n}^{2}t)$ in $r+1$
sums depending on how many components of the $\Lambda_{n} v$ are zero. Pick
$n_{0}$ sufficiently large so that 
$$
\left|a_{i}\right|/2\leq\left|(\Lambda_{n}v)_{i}\right|/u_{n}\leq2\left|a_{i}\right|
$$
for all $v\in\mathbb{Z}^{r}$ and $n\geq n_{0}$ and $a_{i}\neq0$.
Recall that $I_{-n}=I_{n}.$ Let us look at a term with $k$ zeros
in the $y_{i}$s and estimate with the help of Proposition \ref{propestimates}:
\begin{align*}
 u_{n}^{r}e^{-2ru_{n}^{2}t}{\mathbf I}_{v,\Lambda_{n}}(2u_{n}^{2}t)
  &\leq\left(\frac{1}{\sqrt{2t}}\right)^{r}{\displaystyle \prod\limits _{(\Lambda_{n}v)_{i}\neq0}}\left(1+
  \frac{\left\vert (\Lambda_{n}v)_{i}\right\vert }{u_{n}n_{0}2t}\right)^{-n_{0}\left\vert (\Lambda_{n}v)_{i}\right\vert /2u_{n}}
\\& \leq\left(\frac{1}{\sqrt{2t}}\right)^{r}{\displaystyle \prod\limits _{a_{i}\neq0}}\lambda^{\left|a_{i}\right|}
\end{align*}
 for all $n$ large and where \[
\lambda:=\left(1+\frac{a}{n_{0}4t}\right)^{-n_{0}/4}<1\]
 and $a$ is the smallest nonzero absolute value of all the entries
in $A\mathbb{Z}^{r}$. The whole theta series is therefore bounded
by $r+1$ sums of a product of convergent geometric series. This shows
that the infinite sum is uniformly convergent and the proof is complete.
\end{proof}

\begin{lemma} \label{lemuniform} Given a sequence
$\{\Lambda_{n}\}$ satisfying $(\det\Lambda_{n})^{-1/r}\Lambda_{n}\rightarrow A$
as above, there is a constant $d>0$ such that for all sufficiently
large $n$\[
\theta_{\Lambda_{n}}(u_{n}^{2}t)\leq1+\sum_{j=1}^{\infty}e^{-dtj}.\]

\end{lemma}

\begin{proof} Let $u_{n}=\det(\Lambda_{n})^{1/r}.$ We have\begin{align*}
\theta_{\Lambda_{n}}(t) & =\sum_{\nu\in\mathbb{Z}^{r}\backslash\Lambda_{n}^{\ast}\mathbb{Z}^{r}}e^{-t\lambda_{v}}=\sum_{\nu\in\mathbb{Z}^{r}\backslash\Lambda_{n}^{\ast}\mathbb{Z}^{r}}e^{-t\left(2r-2\sum_{k=1}^{r}\cos(2\pi v_{k})\right)}\\
 & =1+\sum_{\substack{\nu\in\mathbb{Z}^{r}\backslash\Lambda_{n}^{\ast}\mathbb{Z}^{r}\\
v\neq0}
}{\displaystyle \prod\limits _{k=1}^{r}}e^{-4t\sin^{2}(\pi v_{k})}\end{align*}
 So that \[
\theta_{\Lambda_{n}}(u_{n}^{2}t)=1+\sum_{\substack{\nu\in\mathbb{Z}^{r}\backslash\Lambda_{n}^{\ast}\mathbb{Z}^{r}\\
v\neq0}
}{\displaystyle \prod\limits _{k=1}^{r}}e^{-4tu_{n}^{2}\sin^{2}(\pi v_{k})}\]
 We use the elementary bounds $\sin x\geq x-x^{3}/6$ and $\sin(\pi-x)\geq x-x^{3}/6$
for $x\in\lbrack0,\pi/2]$ and get\begin{align*}
u_{n}^{{}}\sin(\pi v_{k}) & \geq u_{n}\pi v_{k}(1-\pi^{2}v_{k}^{2}/6)>c\pi u_{n}v_{k}\text{ if }v_{k}\leq1/2\\
u_{n}^{{}}\sin(\pi v_{k}) & \geq u_{n}\pi v_{k}(1-\pi^{2}v_{k}^{2}/6)>c\pi u_{n}(1-v_{k})\text{ if }v_{k}>1/2\end{align*}
 for some postitve constant $c$ for all $n$ sufficiently large.
Note also that for every $v\neq0$, the values $u_{n}v_{k}$ range over the
integers times an entry in $A$ as $n\rightarrow\infty$
because of the convergence of $(\det\Lambda_{n})^{-1/r}\Lambda_{n}$ to $A$ in
$\textrm{SL}_{r}(\mathbb{R}).$  We then conclude there is a constant $d>0$ such
that for all sufficiently large $n$\[
\theta_{\Lambda_{n}}(u_{n}^{2}t)\leq1+\sum_{j=1}^{\infty}e^{-dtj}.\]

\end{proof}

Now we can show:

\begin{proposition} \label{prop1infty}With the notation as above
and $u_{n}:=\det(\Lambda_{n})^{1/r}$, we have that\begin{align*}
 & \int_{1}^{\infty}\left(\theta_{\Lambda_{n}}(u_{n}^{2}t)-u_{n}^{r}e^{-2ru_{n}^{2}t}I_{0}(2u_{n}^{2}t)^{r}-1+e^{-u_{n}^{2}t}\right)\frac{dt}{t}\\
 & =\int_{1}^{\infty}(\theta_{A}(t)-1)\frac{dt}{t}-\frac{2}{r}(4\pi)^{-r/2}+o(1)\end{align*}
 as $n \rightarrow \infty$.
\end{proposition}

\begin{proof} Write\begin{align*}
 & \int_{1}^{\infty}\left(\theta_{\Lambda_{n}}(u_{n}^{2}t)-u_{n}^{r}e^{-2ru_{n}^{2}t}I_{0}(2u_{n}^{2}t)^{r}-1+e^{-u_{n}^{2}t}\right)\frac{dt}{t}\\
 & =\int_{1}^{\infty}\left(\theta_{\Lambda_{n}}(u_{n}^{2}t)-1\right)\frac{dt}{t}-\int_{1}^{\infty}u_{n}^{r}e^{-2ru_{n}^{2}t}I_{0}(2u_{n}^{2}t)^{r}\frac{dt}{t}+\int_{1}^{\infty}e^{-u_{n}^{2}t}\frac{dt}{t}.\end{align*}
 In the last row, the third integral clearly goes to zero as $n\rightarrow\infty.$
For the first integral in the same row we have\[
\int_{1}^{\infty}\left(\theta_{\Lambda_{n}}(u_{n}^{2}t)-1\right)\frac{dt}{t}\rightarrow\int_{1}^{\infty}(\theta_{A}(t)-1)\frac{dt}{t}\]
 in view of the pointwise convergence from Proposition \ref{proppointwise}
and the uniform integrable upper bound from Lemma \ref{lemuniform}.

The middle integral \[
\int_{1}^{\infty}u_{n}^{r}e^{-2ru_{n}^{2}t}I_{0}(2u_{n}^{2}t)^{r}\frac{dt}{t}\]
 converges to \[
\int_{1}^{\infty}(4\pi t)^{-r/2}\frac{dt}{t}=\frac{2}{r}(4\pi)^{-r/2}\]
 in view of the heat kernel convergence from Proposition \ref{propestimates}
and Proposition \ref{prophkconv}, so then we may appeal to the
Lebesgue dominated convergence theorem. \end{proof}

Next we show:

\begin{proposition} \label{prop01a}With the notation as above and
$u_{n}:=\det(\Lambda_{n})^{1/r}$, we have that\[
\int_{0}^{1}\left(\theta_{\Lambda_{n}}(u_{n}^{2}t)-u_{n}^{r}e^{-2ru_{n}^{2}t}I_{0}(2u_{n}^{2}t)^{r}\right)\frac{dt}{t}\rightarrow\int_{0}^{1}\left(\theta_{A}(t)-(4\pi t)^{-r/2}\right)\text{ }\frac{dt}{t}\]
 as $n\rightarrow\infty.$ \end{proposition}

\begin{proof} For fixed $t$ we have the pointwise convergence as
$n\rightarrow\infty$\[
\theta_{\Lambda_{n}}(u_{n}^{2}t)-u_{n}^{r}e^{-2ru_{n}^{2}t}I_{0}(2u_{n}^{2}t)^{r}\rightarrow\theta_{A}(t)-(4\pi t)^{-r/2}.\]
 It remains therefore to exhibit uniform (for $n>>1$) integrable
bounds on the integrands. This can be done in the same way as in the
proof of Proposition \ref{proppointwise}. In order to make the bound
obtained there, in terms of $\lambda,$ integrable for $0\leq t\leq1,$
we just need to choose $n_{0}$ large so that $n_{0}a/4>r/2.$ \end{proof}

Finally we recall from \cite{CJK10}:

\begin{proposition} \label{prop01b}For $u\in\mathbb{R}$ we have
the asymptotic formula\[
\int_{0}^{1}(e^{-u^{2}t}-1)\frac{dt}{t}=\Gamma^{\prime}(1)-2\log(u)+o(1)\]
as $u\rightarrow\infty$.

\end{proposition}

We now turn to the proof of our main result, Theorem 1.

In view
of Theorem \ref{thmlogderivatives} we have\[
\log\det\Delta_{DT(\Lambda_{n})}=\det\Lambda_{n}\cdot c_{r}-\int_{0}^{\infty}\left(\theta_{\Lambda_{n}}(t)-\det\Lambda_{n}\cdot e^{-2rt}I_{0}(2t)^{r}-1+e^{-t}\right)\frac{dt}{t}.\]
After the  change of variables $t\rightarrow u_{n}^{2}t$, the
second term becomes
\begin{align*}
 & -\int_{0}^{\infty}\left(\theta_{\Lambda_{n}}(u_{n}^{2}t)-\det\Lambda_{n}\cdot e^{-2ru_{n}^{2}t}I_{0}(2u_{n}^{2}t)^{r}-1+e^{-u_{n}^{2}t}\right)\frac{dt}{t}\\
 & =-\left[\int_{0}^{1}+\int_{1}^{\infty}\right]\left(\theta_{\Lambda_{n}}(u_{n}^{2}t)-\det\Lambda_{n}\cdot e^{-2ru_{n}^{2}t}I_{0}(2u_{n}^{2}t)^{r}-1+e^{-u_{n}^{2}t}\right)\frac{dt}{t}.\end{align*}
In view of Propositions \ref{prop1infty}, \ref{prop01a},
and \ref{prop01b}, this integral equals \begin{align*}
 & -\int_{1}^{\infty}(\theta_{A}(t)-1)\frac{dt}{t}+\frac{2}{r}(4\pi)^{-r/2}-\int_{0}^{1}\left(\theta_{A}(t)-(4\pi t)^{-r/2}\right)\text{ }\frac{dt}{t}\\
 & -\Gamma^{\prime}(1)+2\log(u_{n})+o(1).\end{align*}
 Keeping in mind that $u_{n}=\left(\det\Lambda_{n}\right)^{1/r}$
and identifying the constant terms appearing in the meromorphic continuation
of $-\zeta^{\prime}(0)=\log\det\Delta_{A\mathbb{Z}^{r}\backslash\mathbb{R}^{r}}$, the
main theorem is proved; cf. equation (15) of \cite{CJK10} with $V(A)=\det A=1$.

\section{Proof of the corollaries}

To prove the corollaries in the introduction we recall the statements
from the literature that we use.

For Corollary \ref{coroptimal} note that the height has a global
minimum for the hexagonal lattice $A_{2}$ in dimension 2 as is well-known
and for the f.c.c. lattice $A_{3}\cong D_{3}$ in dimension 3 by the rigorous
numerics of Sarnak-Str{\"o}mbergsson in \cite{SS06}. Hence for any unimodular
lattice $L$ in the respective dimensions\begin{align*}
h(L) & \geq h(A_{2}),\text{ so }\log\det\Delta_{L}\leq\log\det\Delta_{A_{2}}\\
h(L) & \geq h(D_{3}),\text{ so }\log\det\Delta_{L}\leq\log\det\Delta_{D_{3}}.\end{align*}
 As already remarked the two leading terms in the asymptotics in Theorem
\ref{thmmain} are shape-independent, and so Corollary \ref{coroptimal}
follows from these remarks arguing with convergent subsequences in
view of the compactness.

For Corollary \ref{corestimate} note that Corollary 1 on p. 119 in
\cite{SS06} implies that\[
\log\det\Delta_{M}<\gamma-\log4\pi+\frac{2}{r}<-0.95\]
 where $M$ is an $r$-dimensional flat torus of volume 1 and $\gamma$
is Euler's constant $\gamma\approx0.577$. In view of this statement, let us replace
$\Lambda_{n}$ by a convergent subsequence.  Then by  Theorem \ref{thmmain}
and the matrix-tree theorem, we have \[
\tau(DT(\Lambda_{n}))=\frac{\det^{\prime}\Delta_{DT(\Lambda_{n})}}{\det\Lambda_{n}}\leq\frac{\left(\det\Lambda_{n}\right)^{2/r-1}}{4\pi}\exp(\det\Lambda_{n}\cdot c_{r}+\gamma+2/r)\]
 for all sufficiently large $n$.  This concludes the proof of Corollary \ref{corestimate}.  

Finally, it may be of interest to mention another estimate in \cite{SS06}:\[
h(L)\geq4\sqrt{\frac{\pi}{r}}\left(\frac{\sqrt{r/2\pi e}}{m(L)}\right)^{r}(1+o(1))\]
 where $m(L)$ is the length of the shortest non-zero vector in the
lattice $L$. Recall that being the densest regular packing is equivalent
to being the lattice with co-volume 1 which maximizes the length of
the shortest nonzero vector.

\noindent
Gautam Chinta \\
Department of Mathematics \\
The City College of New York \\
Convent Avenue at 138th Street \\
New York, NY 10031\\
U.S.A. \\
e-mail: chinta@sci.ccny.cuny.edu

\vspace{5mm} \noindent
Jay Jorgenson \\
Department of Mathematics \\
The City College of New York \\
Convent Avenue at 138th Street \\
New York, NY 10031\\
U.S.A. \\
e-mail: jjorgenson@mindspring.com

\vspace{5mm}\noindent
Anders Karlsson\\
Section de math\'ematiques\\
Universit\'e de Gen\`eve\\
2-4 rue du Li\`evre\\
1211 Gen\`eve 4\\
Switzerland\\
e-mail: Anders.Karlsson@unige.ch

\end{document}